\documentclass[letterpaper,USenglish,thm-restate,numberwithinsect]{lipics-v2019}

\graphicspath{{./figures/}}
\title{Summarizing Diverging String Sequences, with Applications to Chain-Letter Petitions}
\titlerunning{Summarizing Diverging String Sequences}

\author{Patty Commins}{%
  Department of Computer Science, Carleton College \and
  Department of Mathematics, University of Minnesota}%
{commins.patty@gmail.com}{}{}

\author{David Liben-Nowell}{%
  Department of Computer Science, Carleton College}
{dln@carleton.edu}{}{}

\author{Tina Liu}{%
  Department of Computer Science, Carleton College \and
  Surescripts}%
{tina.jxy.liu@gmail.com}{}{}

\author{Kiran Tomlinson}{%
  Department of Computer Science, Carleton College \and
  Department of Computer Science, Cornell University}%
{kt@cs.cornell.edu}{}{}

\authorrunning{P.~Commins, D.~Liben-Nowell, T.~Liu, and K.~Tomlinson}
\Copyright{Patty Commins, David Liben-Nowell, Tina Liu, and Kiran Tomlinson}



\supplement{Related research data and source code hosted at \url{https://github.com/tomlinsonk/diverging-string-seqs}.}
\nolinenumbers
\hideLIPIcs

\acknowledgements{We thank Jon Kleinberg for extensive discussions, and Anna Johnson, Hailey Jones, Dave Musicant, Layla Oesper, Anna Rafferty, and Ethan Somes for helpful discussions during preliminary or late stages of this project.  This work was supported in part by Carleton College.}

\EventEditors{Inge Li G{\o}rtz and Oren Weimann}
\EventNoEds{2}
\EventLongTitle{31th Annual Symposium on Combinatorial Pattern Matching (CPM 2020)}
\EventShortTitle{CPM 2020}
\EventAcronym{CPM}
\EventYear{2020}
\EventDate{June 17--19, 2020}
\EventLocation{Copenhagen, Denmark}
\EventLogo{}
\SeriesVolume{161}
\ArticleNo{13}

\usepackage{xspace,tikz,fancyvrb,cite}

\DeclareMathOperator{\medoid}{\ensuremath{\mathsf{medoid}}}
\DeclareMathOperator{\error}{err} 
\DeclareMathOperator{\cost}{\mathcal{C}}

\newcommand{\edg}{\ensuremath{\mathsf{EDG}}}

\newcommand{\ed}{\ensuremath{\mathsf{ED}}}
\newcommand{\ad}{\ensuremath{\mathsf{AED}}}
\newcommand{\buildbifurcation}{\textsc{Build\-Bifurcation}\xspace}
\newcommand{\buildtree}{\textsc{BuildTree}\xspace}

\def\set#1{\ensuremath{\{ #1 \}}}
\def\tup#1{\ensuremath{\langle #1 \rangle}}
\def\eps{\ensuremath{\varepsilon}}
\def\L#1#2{\ensuremath{\mathsf{labelseq}_{#1}(#2)}}

\def\comment#1{\em \small // #1}

\begin{document}
\maketitle
\begin{abstract}

  Algorithms to find optimal alignments among strings, or to find a parsimonious summary of a collection of strings, are well studied in a variety of contexts, addressing a wide range of interesting applications.  In this paper, we consider \emph{chain letters,} which contain a growing sequence of signatories added as the letter propagates.  The unusual constellation of features exhibited by chain letters (one-ended growth, divergence, and mutation) make their propagation, and thus the corresponding reconstruction problem, both distinctive and rich.  Here, inspired by these chain letters, we formally define the problem of computing an optimal summary of a set of diverging string sequences.  From a collection of these sequences of names, with each sequence noisily corresponding to a branch of the unknown tree $T$ representing the letter's true dissemination, can we efficiently and accurately reconstruct a tree $T' \approx T$?  In this paper, we give efficient exact algorithms for this summarization problem when the number of sequences is small; for larger sets of sequences, we prove hardness and provide an efficient heuristic algorithm.  We evaluate this heuristic on synthetic data sets chosen to emulate real chain letters, showing that our algorithm is competitive with or better than previous approaches, and that it also comes close to finding the true trees in these synthetic datasets.  

\end{abstract}
   
\newpage
\section{Introduction}
\label{sect:intro}

In a range of computational settings, we are given a collection of strings and asked to construct some kind of parsimonious representation of the given set.  The task becomes more interesting if the strings result from a generative process, especially if new strings arise from a mechanism involving both replication and mutation of old strings.  Now the parsimonious representation might be a \emph{tree} describing the generative history of this population, with nodes
corresponding to the strings and the branching structure representing the evolutionary events that produced that population.  
At this level of description,
a host of applications fall under this rubric:  reconstructing a phylogeny from a set of genes, tracing the spread of a textual meme in a social network, inferring the version history of a document from its many copies.  These domains differ in the way that replication and mutation occur---sometimes randomly (``by nature'') and sometimes intentionally by humans embedded in a social structure---and, perhaps, whether some kind of selective pressure affects which strings survive or replicate.

Here, we consider a specific---and surprisingly rich---social setting in which a population of strings is generated:  \emph{chain letters.}  Chain letters often feature an outlandish claim (``send a copy of this letter to ten friends, or you will have bad luck forever!''), but, more crucially, recipients are instructed to add their names to the document's end, make a copy, and send those copies to multiple friends.  Importantly, every subsequent recipient may modify any part of the document, or copy it imprecisely; thus each document contains a list of signatures, representing an ordered (if noisy) trace of its particular path through the social network.  In the present work, we study the problem of accurately reconstructing the underlying propagation tree from a set of signature lists.  If the lists of signatures contained no errors, the problem would be trivial, but errors abound, including both point mutations and structural variants.  (In real email-based chain-letter data, some signatories retyped names, often incorrectly, and both block deletions and duplications appear~\cite{LNK08}. Worse, some copies of the emails are only available as low-quality scanned images,
introducing further errors.)

\subparagraph*{The propagation of chain letters.}

There are three crucial properties in chain-letter--like contexts that, together, make this data
intriguingly different from other settings:
\begin{romanenumerate}
\item \emph{chain letters grow (at one end).}  A document has an ``active end,'' and a document typically changes via the deposition of additional text (another name) at its active end.
  
\item \emph{chain letters diverge.}  A document can \emph{split} to create multiple ``children'' documents, which share a prefix up to the split but have differing suffixes below.  The split is at the active end; two documents that diverge grow independently after the branching point.

\item \emph{chain letters mutate.}  Actors introduce noise: an individual sending a chain letter to a friend makes a (potentially imperfect) \emph{copy} of that document, possibly introducing errors---and those errors are ``inherited'' by subsequent copies of the letter.
\end{romanenumerate}
Given a collection of many copies of ``the same'' chain letter, each with its own sequence of names, one can seek to reconstruct the underlying true record of the propagation---both the structure of the propagation tree \emph{and} the strings representing the true names of the signatories.  Together, the above properties make this chain-letter reconstruction problem a tantalizing domain for parsimonious reconstruction:  as the rate of noise in document copying increases, naturally the reconstruction problem becomes difficult, but there is a great deal of repetition in the input data, particularly near the root of the propagation tree.

\subparagraph*{The present work.}
We formally introduce the \emph{Diverging String Sequence Summarization Problem (DSSSP):}   given a collection $X$ of sequences of strings (strings correspond to names, and each sequence is a noisy list of names in an instantiation of a chain letter), we seek a tree $T$ that optimally summarizes $X$.  (Rather than using the language of chain letters, we will abstract away the particular application and discuss diverging string sequences in general.)%
\footnote{There is another layer of complication, literally:  rather than viewing a document as a sequence of \emph{characters,} we instead view it as a sequence of \emph{signatures} (each of which is a string that consists of a sequence of characters).  Thus there is a ``two-level'' view of edits, in which either an individual character can be corrupted (a single character-level edit within a particular signature) or an individual signature can be corrupted (an entire signature is deleted, inserted, or replaced by a different signature).} %
What counts as an ``optimal'' summary depends on a tradeoff between two competing goods:  the \emph{accuracy} of $T$ in representing the strings in the given sequences, and the \emph{efficiency} of $T$ in representing the given string sequences without too much redundancy.  Our formal definition of the problem is parameterized to reflect the tradeoff between these two competing goods.

Our main theoretical results on DSSSP are (1) an efficient optimal algorithm for the case of $m = 2$ sequences, based on an approach we call \emph{edit distance with give-up} (Theorem~\ref{thm:optimality-of-two-string-reconstruction}); (2) a proof of hardness for large $m$ (Theorem~\ref{thm:hardness-of-reconstruction}); and (3) an exact polynomial-time algorithm for any fixed value of $m$ (Theorem \ref{thm:dsssp_is_fixed_parameter_tractable}).  We also give a much more efficient heuristic algorithm for large $m$---using a combination of divergence-aware pairwise alignment and iterative merging, inspired by progressive alignment algorithms~\cite{FD87}---and show empirically that it does a good job of reconstructing synthetically generated trees.

\section{Related Work}

\subparagraph*{Chain-letter data.}  

In joint work with Jon Kleinberg, the second author studied the propagation of a widespread email-based anti-war petition
~\cite{LNK08}.  This work focused on the topological structure of the underlying propagation tree; subsequent research sought to explain the shape of the tree through stochastic branching processes~\cite{GJ10} or the rarity of sampled email copies~\cite{CLNK11}.  The present work differs in that here we study the problem of accurately reconstructing the propagation tree from signature lists, rather than seeking to understand the structure of that tree.  Still, examining the structure of the propagation tree presupposes a reconstructed tree, which in~\cite{LNK08} was done 
using a hard edit distance cutoff to decide whether two signatures belong to the same signatory.  (See Section~\ref{sect:comparing-reconstruction-to-ground-truth}.)
This specific aspect of our problem---do multiple signatures belong to the same signatory?---has been considered in other forms in the past, including error-tolerant recognition of strings 
with various error models~\cite{oflazer1996error-tolerant,brill2000an_improved_error_model}, error correction of strings of regular languages~\cite{wagner1974order-n}, and block edit models for approximate string matching~\cite{lopresti1997block}, all of which use various versions of edit distance.

Chain letters in paper form have also been investigated in the context of
constructing a phylogeny based on variations in the text of the document itself (rather than a list of signatories)~\cite{bennett2003chain}, or the propagation of stories as a network~\cite{karsdorp2016structure}.

\subparagraph*{Other forms of propagation.}  
In rare cases, a situation matching all three key features of chain letters has been studied---including a (controversial) model of the origin of life, based on layered clay accreting over time and even diverging and mutating~\cite{cairns1990seven,bullard2007test}.
More common settings share two of the three features.
For example, absent any errors, our reconstruction task is solved by a trie~\cite{DeLaBriandais1959, Fredkin1960} summarizing a set of diverging strings.  Online conversations~\cite{KMM10:dynamics_of_conversations} (e.g., comment threads or especially email threads) have an active end at which new contributions appear, and threads can diverge, but there is no obvious notion of mutation.

There are also applications in which the objects of interest are strings that grow at one end, with noise but without meaningful divergence.  In dendrochronology (the science of dating wood), approaches based on edit distance can be used to study sequences of growth rings in trees, which accumulate on one end, adjacent to the bark~\cite{wenk1999applying}.

By far, though, the best-studied settings that match two of our three features have strings that mutate and replicate, but have no ``active end'' at which growth occurs.  This is the classical setting of phylogenetic reconstruction, but it also appears in many other contexts.  Most prominent is the spread of news, memes, and rumors that evolve as versions are created and shared~(e.g., \cite{ALAN:wsdm16:info-evo-in-soc-nets,leskovec2009meme,simmons2011memes,horta2019message,friggeri2014rumor}). Mutations in these cases differ from ours, though, in that the content of the information being spread can affect the type of mutations that occur, thereby affecting the likelihood of further propagation (and therefore the structure of the tree).  Much of this work
seeks to understand various types of dissemination and what factors may impact the propagation structure---different from the goal of reconstructing the underlying tree.  Reconstruction of the evolutionary history of a collection of divergent objects is also well studied in a bafflingly wide variety of contexts, from version histories of code snippets in Stack Overflow~\cite{baltes2019sotorrent},
to variations of the story ``Little Red Riding Hood''~\cite{tehrani2013phylogeny}, to 
diverging cultural histories using textile data~\cite{matthews2011testing}.

\subparagraph*{Reconstruction algorithms.}
In addition to the algorithmic approaches to these various other forms of data, there is a voluminous literature on string alignment in the computational biology literature.  The multiple sequence alignment problem is closely related to DSSSP, and many algorithms target a variety of challenges related to it~(see \cite{li1992a_survey,sankoff1975minimal,waterman1976some,raphael2004novel,carrillo1988multiple,kececioglu1993maximum}, among many others). There is also work involving the alignment of amino acid sequences to reconstruct the history of proteins, including mutations and divergence events~\cite{doolittle1992reconstructing}.


\section{Summarizing Diverging String Sequences}
\label{sect:the-problem}

Before we formally define our abstract problem, we begin with some intuition, with terminology drawn from chain letters.  Informally, a \emph{name} is a string over a finite alphabet, and a \emph{petition} is a sequence of names.  We are given a \emph{set} of petitions $X = \set{x_1, \ldots, x_m}$, and we seek the tree $T$ that best summarizes the set $X$.  But the ``best'' tree depends on a tradeoff between two competing goods:  (i) the efficiency of $T$ (its number of nodes), and (ii) the accuracy of $T$ in representing the petitions in $X$.

Consider petitions $x_1 = \text{\Verb|Aaa Bbb Ccc Ddd Eee|}$ and 
$x_2 = \text{\Verb|Aaa Bbx Ccc Dxx Fff|}$, with spaces separating names, as an example.
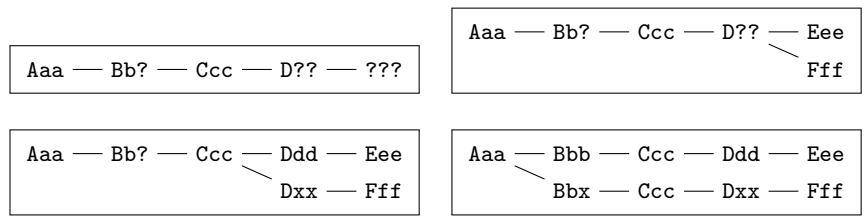
\begin{figure}[b]
  \begin{center}\def\ttt{\tt \scriptsize}%
\tikzset{every picture/.style={baseline=(current bounding box.north),
    xscale=0.45,every node/.style={inner sep=1pt,outer sep=1pt}}}%
%
\def\ttt{\tt}%
\tikzset{every picture/.style={xscale=0.75}}
%
\begin{tabular}{ll}
\fbox{\begin{tikzpicture}
\node (v1) at (-4,1.5) {\ttt Aaa};
\node (v2) at (-2.5,1.5) {\ttt Bb?};
\node (v3) at (-1,1.5) {\ttt Ccc};
\node (v4) at (0.5,1.5) {\ttt D??};
\node (v5) at (2,1.5) {\ttt ???};
\draw (v1) edge (v2);
\draw (v2) edge (v3);
\draw (v3) edge (v4);
\draw (v4) edge (v5);
\end{tikzpicture}}
&
\fbox{\begin{tikzpicture}
\node (v1) at (-4,1.5) {\ttt Aaa};
\node (v2) at (-2.5,1.5) {\ttt Bb?};
\node (v3) at (-1,1.5) {\ttt Ccc};
\node (v4) at (0.5,1.5) {\ttt D??};
\node (v5) at (2,1.5) {\ttt Eee};
\node (v9) at (2,1) {\ttt Fff};
\draw  (v1) edge (v2);
\draw  (v2) edge (v3);
\draw  (v3) edge (v4);
\draw  (v4) edge (v5);
\draw  (v4) edge (v9);
\end{tikzpicture}}
\\~\\ 
\fbox{\begin{tikzpicture}
\node (v1) at (-4,1.5) {\ttt Aaa};
\node (v2) at (-2.5,1.5) {\ttt Bb?};
\node (v3) at (-1,1.5) {\ttt Ccc};
\node (v4) at (0.5,1.5) {\ttt Ddd};
\node (v5) at (2,1.5) {\ttt Eee};
\node (v8) at (0.5,1) {\ttt Dxx};
\node (v9) at (2,1) {\ttt Fff};
\draw  (v1) edge (v2);
\draw  (v2) edge (v3);
\draw  (v3) edge (v4);
\draw  (v4) edge (v5);
\draw  (v3) edge (v8);
\draw  (v8) edge (v9);
\end{tikzpicture}}
&
\fbox{\begin{tikzpicture}
\node (v1) at (-4,1.5) {\ttt Aaa};
\node (v2) at (-2.5,1.5) {\ttt Bbb};
\node (v3) at (-1,1.5) {\ttt Ccc};
\node (v4) at (0.5,1.5) {\ttt Ddd};
\node (v5) at (2,1.5) {\ttt Eee};
\node (v6) at (-2.5,1) {\ttt Bbx};
\node (v7) at (-1,1) {\ttt Ccc};
\node (v8) at (0.5,1) {\ttt Dxx};
\node (v9) at (2,1) {\ttt Fff};
\draw  (v1) edge (v2);
\draw  (v2) edge (v3);
\draw  (v3) edge (v4);
\draw  (v4) edge (v5);
\draw  (v1) edge (v6);
\draw  (v6) edge (v7);
\draw  (v7) edge (v8);
\draw  (v8) edge (v9);
\end{tikzpicture}}
\end{tabular}

\end{center}
  \caption{Four ``best'' trees for $x_1 = \text{\Verb|Aaa Bbb Ccc Ddd Eee|}$ and $x_2 = \text{\Verb|Aaa Bbx Ccc Dxx Fff|}$.}
  \label{fig:best-trees-informally}
\end{figure}
Depending on the relative importance of efficiency and accuracy, there are \emph{four} distinct ``best'' trees (see Figure~\ref{fig:best-trees-informally}):  a trivial tree that never diverges (if efficiency matters much more than accuracy); a tree that diverges upon any textual discrepancy (if accuracy matters much more); or two intermediate trees that diverge after the \Verb|Ccc|s or the \Verb|D??|s (depending on the cost--benefit of adding one node vs.~paying for two textual errors).  

\subsection{Distance between a Summary Tree and a String Sequence}

To begin, we need to quantify how accurately a set $X$ of string sequences is represented by a summary structure---and, more fundamentally, what it means to summarize $X$.
\begin{definition}[Labeled Summary Tree]
  Let $X$ be a set of string sequences.  A \emph{labeled summary tree} of $X$ is a pair $\tup{T,f}$, where $T$ is a tree with each node labeled with a string, and $f$ is a function mapping each $x \in X$ to a node $v_x$ in $T$.  

  Let $\L{T}{v_x}$ denote the sequence of node labels on the path from the root of $T$ to $v_x$.
\end{definition}
That is, a summary of $X$ consists of a labeled tree $T$, with a node of $T$ designated to correspond to each sequence in $X$.
To assess how accurately a string sequence $x \in X$ is represented, we will compare $x$ with $\L{T}{v_x}$.  Our metric will be a variation on the classical Levenshtein edit distance~\cite{Lev66}, with two adjustments:
\begin{enumerate}
\item Edit distance is usually defined between two strings, but we wish to compare two \emph{sequences} of strings. We cannot simply concatenate the strings within each sequence and then use standard edit distance, as the edits would no longer respect string boundaries. As such, we need to define an edit distance with two levels of granularity.

\item We insist that every string in each sequence $x \in X$ be represented (perhaps with some error) in the tree.  Thus, when we align $x$ to its path in the tree, we do not allow deletion of strings in the sequence.  (We forbid deletions from $x \in X$ to preserve the intuition that the optimal summary tree for a singleton sequence $X = \set{x}$ is a nonbranching path successively labeled by the strings in $x$ \emph{even when the cost of nodes is very high}.)  
\end{enumerate}
Let $x$ be a string sequence, and let $y = \L{T}{v_x}$.  Our metric, then, is an asymmetric two-level variant of the edit distance between $x$ and $y$, allowing deletions from $y$ but not $x$:
\begin{definition}[Asymmetric Edit Distance]
  Let $x$ and $y$ be string sequences.  The \emph{asymmetric edit distance} of $x$ with respect to $y$, denoted $\ad(x,y)$, is the cost of the cheapest sequence of operations transforming $x$ into $y$, where the allowable operations are inserting a string into $x$ and substituting a string.  (No string can be deleted from $x$.)

  These operations' costs are given by (classical) edit distance $\ed$: substituting $w'$ for $w$ costs $\ed(w,w')$; inserting a string $w$ into $x$ costs $\ed(w, \varepsilon)$, where $\varepsilon$ is the empty string.  (Unless otherwise specified, all $\ed$ edits have unit cost, but we allow arbitrary cost matrices.)
\end{definition}

We compute $\ad(x,y)$ with a variation on the classical dynamic program for edit distance,
forbidding deletions 
and using $\ed$ to compute the cost of inserting or substituting a string.  

The \emph{distance} between a string sequence $x$ and summary tree $\tup{T,f}$, then, is given by $\ad(x, \L{T}{f(x)})$---i.e., the asymmetric edit distance between $x$ and the label sequence on the path from root to the node corresponding to $x$ in $T$.

\subsection{The Problem:  Summary Trees for a Set of String Sequences}
We can now formally define our problem, where our objective function is---in the style of regularization in machine learning~\cite{hastie2009elements}---a weighted sum of the accuracy of the summary tree (as measured by \ad) and the simplicity of the tree (as measured by its number of nodes):

\begin{definition}[Diverging String Sequence Summarization Problem {[DSSSP]}]~
  \label{def:dsssp}
  \begin{description}
  \item[Input:] A set $X = \set{x_1, x_2, ..., x_m}$ of string sequences and a nonnegative \emph{node cost} $\lambda$.
  \item[Output:] A labeled summary tree $\tup{T,f}$ (i.e., a tree $T$ and a function $f$ mapping each $x_i$ to a node $v_i$ in $T$) minimizing the following, where $|T|$ denotes the number of nodes in $T$:
    \begin{equation}\label{eq:error}
      \error_{\lambda}(T) :=
       \Big[\sum\nolimits_{i = 1}^{m} \ad(x_i, \L{T}{v_i})\Big] + \lambda \cdot |T|.
    \end{equation}
  \end{description}
\end{definition}
To ensure that $T$ is a tree with a single root, we place sentinel values at the start of each sequence $x_i$.  (Denote by $|T|$ the number of \emph{non-sentinel} nodes in $T$.)
Note that $\ad(x_i, \L{T}{v_i})$ is defined only if $|x_i| \le |\L{T}{v_i}|$---that is, the depth of the node $v_i$ in $T$ is at least the number of strings in the sequence $x_i$---as it would otherwise be impossible to convert $x_i$ into $\L{T}{v_i}$ without deleting strings.

The parameter $\lambda$ controls the tradeoff between trees that represent the input sequences accurately and trees that provide more concise summaries of the input set.
When $\lambda = 0$, a trivial branch-immediately tree $T$ is optimal; when $\lambda = \infty$, a trivial never-branching tree of depth $\max_i |x_i|$ is.  Intermediate values of $\lambda$ give more interesting structures. See 
Figure~\ref{fig:tree_and_petitions}.

\begin{figure}
    \centering
\def\s{\scriptsize}%
{\scriptsize
  \begin{subfigure}[b]{0.15\textwidth}
    \begin{minipage}[b]{\linewidth}\begin{tabular}[t]{@{}c@{~}c@{~}c@{\hspace*{-1pc}}}
                                            $x_1$ & $x_2$ & $x_3$ \\
       \begin{tabular}[t]{|@{\hskip 1pt}l@{\hskip 1pt}|}\hline \s Alice\\\s Bot\\\s Carol\\\s Eve\\\hline\end{tabular}&
       \begin{tabular}[t]{|@{\hskip 1pt}l@{\hskip 1pt}|}\hline \s Alice\\\s Bob\\\s Carl\\\s Frank\\\hline\end{tabular}&
       \begin{tabular}[t]{|@{\hskip 1pt}l@{\hskip 1pt}|}\hline \s Alyce\\\s Bob\\\s Dan\\\hline\end{tabular}\\
    \end{tabular}\hspace*{-0.5pc}  
  \end{minipage}
  \caption{\raggedright Three string sequences to summarize.}
\end{subfigure}
\hfill
\begin{subfigure}[b]{0.15\textwidth}
\begin{tikzpicture}[yscale=0.6,xscale=0.5]
\node [inner sep=0pc] (v0) at (4,0.75) {\Large $\bullet$};    
\node [draw,minimum width=0.9cm] (v1) at (4,0) {Alice};
\node [draw,minimum width=0.9cm] (v2) at (4,-1) {Bob};
\node [draw,minimum width=0.9cm] (v3) at (4,-2) {Carl};
\node [draw,minimum width=0.9cm] (v5) at (4,-3) {Even};
\draw [->] (v0) edge (v1.north);
\draw [->] (v1) edge (v2.north);
\draw [->] (v2) edge (v3.north);
\draw [->] (v3) edge (v5.north);
\node [left of=v3,xshift=1.2pc,anchor=east] {$v_3$};
\node [left of=v5,xshift=1.2pc,anchor=east] {$v_1, v_2$};
\end{tikzpicture}
\caption{\raggedright An optimal tree for $\lambda \ge 5$.\label{fig:ex:giant-lambda}}
\end{subfigure}
\hfill
\begin{subfigure}[b]{0.18\textwidth}
  \centering
\begin{tikzpicture}[yscale=0.6,xscale=0.5]
\node [inner sep=0pc] (v0) at (4,0.75) {\Large $\bullet$};    
\node [draw,minimum width=0.9cm] (v1) at (4,0) {Alice};
\node [draw,minimum width=0.9cm] (v2) at (4,-1) {Bob};
\node [draw,minimum width=0.9cm] (v3) at (4,-2) {Carl};
\node [draw,minimum width=0.9cm] (v5) at (3,-3) {Eve};
\node [draw,minimum width=0.9cm] (v6) at (5,-3) {Frank};
\draw [->] (v0) edge (v1.north);
\draw [->] (v1) edge (v2.north);
\draw [->] (v2) edge (v3.north);
\draw [->] (v3) edge (v5.north);
\draw [->] (v3) edge (v6.north);
\node [right of=v3,xshift=-1.4pc,yshift=0.2pc, anchor=west] {$v_3$};
\node [above left of=v5,xshift=0.8pc,yshift=-0.9pc] {$v_1$};
\node [above right of=v6,xshift=-0.8pc,yshift=-0.9pc] {$v_2$};
\end{tikzpicture}
\caption{\raggedright The optimal tree for $3 < \lambda < 5$.\label{fig:ex:big-lambda}}
\end{subfigure}
\hfill
\begin{subfigure}[b]{0.215\textwidth}
  \begin{tikzpicture}[yscale=0.6,xscale=0.5]
\node [inner sep=0pc] (v0) at (3,0.75) {\Large $\bullet$};    
\node [draw,minimum width=1cm] (v1) at (3,0) {Alice};
\node [draw,minimum width=0.9cm] (v2) at (3,-1) {Bob};
\node [draw,minimum width=0.9cm] (v3) at (2,-2) {Carol};
\node [draw,minimum width=0.9cm] (v4) at (4,-2) {Dan};
\node [draw,minimum width=0.9cm] (v5) at (1,-3) {Eve};
\node [draw,minimum width=0.9cm] (v6) at (3,-3) {Frank};
\draw [->] (v0) edge (v1.north);
\draw [->] (v1) edge (v2.north);
\draw [->] (v2) edge (v3.north);
\draw [->] (v2) edge (v4.north);
\draw [->] (v3) edge (v5.north);
\draw [->] (v3) edge (v6.north);
\node [above left of=v5,xshift=0.8pc,yshift=-0.9pc] {$v_1$};
\node [right of=v6,xshift=-0.8pc] {$v_2$};
\node [right of=v4,xshift=-1.4pc,yshift=0.2pc, anchor=west] {$v_3$};
\end{tikzpicture}
\caption{\raggedright An optimal tree for $1 \le \lambda \le 3$.\label{fig:ex:medium-lambda}}
\end{subfigure}%
%
%
\begin{subfigure}[b]{0.2675\textwidth}%
\begin{tikzpicture}[yscale=0.6,xscale=0.5]%
\node [inner sep=0pc] (v0) at (4,0.75) {\Large $\bullet$};    
\node [draw,minimum width=0.9cm] (va1) at (2,0) {Alice};
\node [draw,minimum width=0.9cm] (va2) at (2,-1) {Bot};
\node [draw,minimum width=0.9cm] (va3) at (2,-2) {Carol};
\node [draw,minimum width=0.9cm] (va5) at (2,-3) {Eve};
\node [draw,minimum width=0.9cm] (vb1) at (4,0) {Alice};
\node [draw,minimum width=0.9cm] (vb2) at (4,-1) {Bob};
\node [draw,minimum width=0.9cm] (vb3) at (4,-2) {Carl};
\node [draw,minimum width=0.9cm] (vb5) at (4,-3) {Frank};
\node [draw,minimum width=0.9cm] (vc1) at (6,0) {Alyce};
\node [draw,minimum width=0.9cm] (vc2) at (6,-1) {Bob};
\node [draw,minimum width=0.9cm] (vc3) at (6,-2) {Dan};
\draw [->] (v0) edge (va1.north);
\draw [->] (v0) edge (vb1.north);
\draw [->] (v0) edge (vc1.north);
\draw [->] (va1) edge (va2.north);
\draw [->] (va2) edge (va3.north);
\draw [->] (va3) edge (va5.north);
\draw [->] (vb1) edge (vb2.north);
\draw [->] (vb2) edge (vb3.north);
\draw [->] (vb3) edge (vb5.north);
\draw [->] (vc1) edge (vc2.north);
\draw [->] (vc2) edge (vc3.north);
\node [right of=vb5,xshift=-0.8pc] {$v_2$};
\node [right of=vc3,xshift=-0.8pc] {$v_3$};
\node [left of=va5,xshift=0.9pc] {$v_1$};
\end{tikzpicture}%
\caption{\raggedright An optimal tree for $\lambda = 0$.\label{fig:ex:small-lambda}}%
\end{subfigure}%
}%
\hfill%

    \caption{A set of string sequences and their optimal summary trees, for different ranges of $\lambda$.  The $\bullet$ root node denotes the sentinel string starting every sequence; nodes marked by $\set{v_1, v_2, v_3}$ correspond to the sequences $\set{x_1, x_2, x_3}$. The choice about whether to split Eve and Frank into two nodes (Figure~\ref{fig:ex:giant-lambda} vs.~\ref{fig:ex:big-lambda}) is a function of their edit distance, $\ed(\mathtt{Eve},\mathtt{Frank}) = 5$. When $\lambda > 5$, then it is cheaper to accept the cost of aligning both to a single label than to pay for an extra node.}
    \label{fig:tree_and_petitions}
\end{figure}
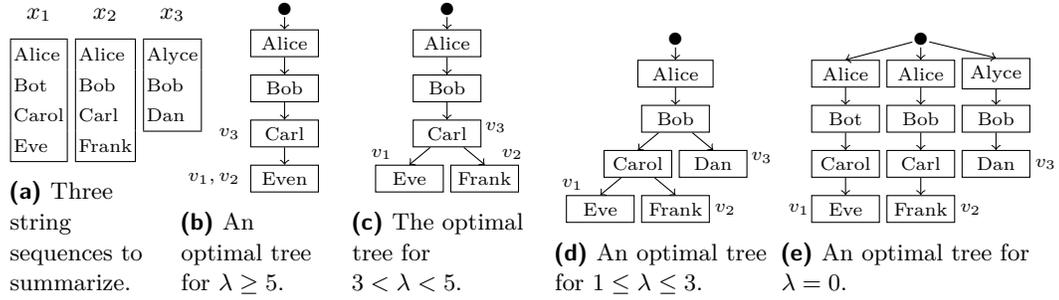

\section{Solving DSSSP for Two Sequences: Edit Distance with Give-up}

Consider first the case of just two input sequences, $m = |X| = 2$.
(Even with $m = 2$, the problem has interesting subtleties.)
We can compute an optimal tree through an alignment algorithm we call \emph{edit distance with give-up.}  The resulting tree has exactly one leaf or two leaves; in the latter case, we call the tree a \emph{bifurcation}.

As with $\ad$, the idea is similar to the classical dynamic program for edit distance, but with one additional operation permitted:  \emph{give up} entirely on aligning the remaining portions of the sequences, and declare a split at this point.  We also modify the costs in the edit distance dynamic program to reflect the $\lambda$ per-node cost of each operation, corresponding to the node-cost term in $\error_\lambda(T)$. 
Writing $\edg(i,j,\lambda)$ to denote the cost of the best alignment of $x_{i, \ldots, |x|}$ and $y_{j, \ldots, |y|}$ under node cost $\lambda$, and writing $\edg(x,y,\lambda) = \edg(1,1,\lambda)$, we have
\begin{align}\label{eq:edg-recurrence}
   \edg(|x|+1, |y|+1,\lambda) &= 0 \\
   \edg(i, |y|+1, \lambda) &= \lambda(|x|-i + 1) \qquad \text{for any $0 \le i \le |x|$}\nonumber\\  
  \edg(|x|+1, j, \lambda) &= \lambda(|y|-j + 1) \qquad \text{for any $0 \le j \le |y|$} \nonumber\\
  \intertext{and, for any $0 \le i \le |x|$ and any $0 \le j \le |y|$,}
   \edg(i,j,\lambda) &= \text{min}  \nonumber
	\left\{\begin{array}{l@{~+~}lll}
    	\edg(i+1, j+1,\lambda) & \lambda + \ed(x_i, y_j) &\text{(substitution)}\\
    	\edg(i, j+1,\lambda) &\lambda + \ed(\varepsilon, y_j)&\text{(insertion)}\\
    	\edg(i+1, j,\lambda) & \lambda + \ed(x_i, \varepsilon) &\text{(deletion)}\\
        \lambda (|x|-i + 1)& \lambda(|y|-j + 1)&\text{(give up)}
        \end{array}\right.
\end{align}
For example, consider the insertion case of the minimum.  Here we match the string $y_j$ with no corresponding entry in $x$, and recursively align $y_{j+1, \ldots, |y|}$ with $x_{i, \ldots, |x|}$, with a total cost of
\[
  \underbrace{\strut\edg(i, j+1,\lambda)}_{\text{cost of alignment of remaining strings}}
  + \underbrace{\lambda\strut}_{\text{cost of creating the root node}}
  + \underbrace{\strut\ed(\varepsilon, y_j).}_{\text{cost of inserting $y_j$ into $\L{T}{x}$}}
\]
The cost of ``giving up''---i.e., declaring a split in the alignment---is large, requiring $(|x|-i + 1)$ nodes on the $x$ branch and $(|y|-j + 1)$ on the $y$ branch, each of which incurs cost $\lambda$.

\begin{figure}[t]
  \centering
  \tikzset{pathnode/.style={draw, circle,inner sep=0pc, outer sep=0pc}}
  \tikzset{every picture/.style={baseline=(current bounding box.north),
      xscale=0.75,every node/.style={inner sep=0pt,outer sep=0pt, minimum size=0.5pc}}}%
  \def\rootcontinue#1#2#3{%
      \begin{tikzpicture}
        \node [draw, circle, inner sep=0pc,outer sep=0pc] (v1) at (0,0) {};
        \node [draw=none, right of=v1, xshift = -1.9pc, anchor=west] {\scriptsize #1};
        \node [draw, inner sep=0.5pc,outer sep=0pc] (v2) at (0,-1) {\scriptsize $\mathrm{BB}(#2,#3,\lambda)$};
        \node [draw=none] (v5) at (0,0.5) {};
        \draw [->] (v5) edge (v1);
        \draw [->] (v1) edge (v2);
      \end{tikzpicture}%
}
\def\xysplit{
    \begin{tikzpicture}
        \node [draw=none] (v0) at (0,0.5) {};
        \node [draw, circle, inner sep=0pc,outer sep=0pc] (v1) at (-1.0,0) {};
        \node [draw=none, right of=v1, xshift = -1.9pc, anchor=west] {\scriptsize $x_1$};
        \node [draw, circle, inner sep=0pc,outer sep=0pc] (v2) at (-1.0,-0.5) {};
        \node [draw=none, right of=v2, xshift = -1.9pc, anchor=west] {\scriptsize $x_2$};
        \node (v3) at (-1.0,-1) {\tiny $\vdots$};
        \node [draw, circle, inner sep=0pc,outer sep=0pc] (v4) at (-1.0,-1.5) {};
        \node [draw=none, right of=v4, xshift = -1.9pc, anchor=west] {\scriptsize $x_{|x|}$};
        \draw [->] (v0) edge (v1);
        \draw [->] (v1) edge (v2);
        \draw (v2) edge (v3);
        \draw [->] (v3) edge (v4);
        \node [draw, circle, inner sep=0pc,outer sep=0pc] (w1) at (1.0,0) {};
        \node [draw=none, right of=w1, xshift = -1.9pc, anchor=west] {\scriptsize $y_1$};
        \node [draw, circle, inner sep=0pc,outer sep=0pc] (w2) at (1.0,-0.5) {};
        \node [draw=none, right of=w2, xshift = -1.9pc, anchor=west] {\scriptsize $y_2$};
        \node (w3) at (1.0,-1) {\tiny $\vdots$};
        \node [draw, circle, inner sep=0pc,outer sep=0pc] (w4) at (1.0,-1.5) {};
        \node [draw=none, right of=w4, xshift = -1.9pc, anchor=west] {\scriptsize $y_{|y|}$};
        \draw [->] (v0) edge (w1);
        \draw [->] (w1) edge (w2);
        \draw (w2) edge (w3);
        \draw [->] (w3) edge (w4);
      \end{tikzpicture}}
  
\begin{tabular}{c|c|c|c}
  $x = \eps$ or $y_1$ is inserted &
  $y = \eps$ or $x_1$ is inserted & 
  substitute $x_1$ for $y_1$ &
  give up (= diverge)\\[-1pc]
  \rootcontinue{$y_1$}{x_{1\ldots|x|}}{y_{2\ldots|y|}} &
  \rootcontinue{$x_1$}{x_{2\ldots|x|}}{y_{1\ldots|y|}} &
  \rootcontinue{$x_1$ or $y_1$}{x_{2\ldots|x|}}{y_{2\ldots|y|}} &
  \xysplit
\end{tabular}
\caption{Constructing the tree in $\buildbifurcation(x,y,\lambda)$. Whenever a node corresponding to $x_{|x|}$ or $y_{|y|}$ is placed, we map the corresponding input sequence to that node of the bifurcation. We write ``BB'' to abbreviate $\buildbifurcation$, and we output an empty tree when $x = y = \eps$.}
\label{fig:bb_bifurcation-itself}
\end{figure}
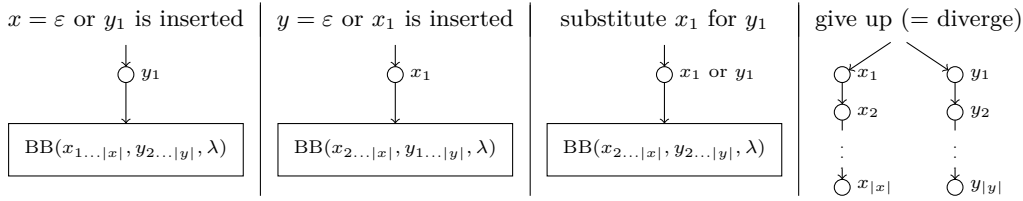

Denote by \buildbifurcation the natural dynamic programming algorithm that computes $\edg$ using (\ref{eq:edg-recurrence}).  (We abuse notation: $\buildbifurcation(x,y,\lambda)$ denotes either the resulting bifurcation or its cost.  We can construct this bifurcation simultaneously with the construction of the alignment; see Figure~\ref{fig:bb_bifurcation-itself}.)

$\buildbifurcation$ optimally solves DSSSP for $m = 2$ sequences---i.e., the tree $T$ built by $\buildbifurcation(x, y, \lambda)$ minimizes $\error_\lambda(T)$. %
(The proof is deferred to Appendix~\ref{sect:proofs}.)

\begin{restatable}{theorem}{bbisoptimal}
  \label{thm:opt} \label{thm:optimality-of-two-string-reconstruction}
The tree $T^\ast := \buildbifurcation(x,y,\lambda)$ is an optimal summary tree for the string sequences $\set{x,y}$ with node cost $\lambda$.  Specifically, $\error_\lambda(T^\ast) = \edg(x, y, \lambda)$.
\end{restatable}
\noindent%
Writing $n = \max(|x|,|y|)$ to denote the length of the longer of the two sequences, and $k = \max(\max_i |x_i|, \max_j |y_j|)$ to denote the length of the longest string in either sequence, then the running time of $\buildbifurcation(x, y, \lambda)$ is $O(n^2 k^2)$.
    
\section{Optimal Summaries of Larger Sets of String Sequences}
\label{sect:larger-sets}

$\buildbifurcation$ efficiently finds the optimal summary tree for $m = 2$ sequences, but DSSSP with large $m$ is computationally intractable.  

\begin{restatable}{theorem}{dssspishard}
\label{thm:hardness-of-reconstruction}
  DSSSP (for an arbitrary number $m$ of string sequences) is NP-hard.
\end{restatable}%
 \begin{proof}[Proof (idea).]
   For large $\lambda$, hardness follows from a reduction from String Median, which we will encounter shortly.

   While the reduction from Median String is simpler, we can also give an alternative reduction from Shortest Common Supersequence (SCS)~\cite{RAIHA1981187}, which applies for smaller $\lambda$ as well. (Note that SCS is hard in general, but for a small number of strings it is efficiently solvable~\cite{fraser1995subsequences}.) The details of the latter proof are in the appendix.
\end{proof}%
\noindent%
On the other hand, if we are willing to tolerate running times that are exponential in $m$ (but polynomial in the other measures of input size), we can solve DSSSP in polynomial time:
\begin{theorem}
  \label{thm:dsssp_is_fixed_parameter_tractable}
  Let $X$ be a set of string sequences, where $m = |X|$ is the number of sequences, $n = \max_i |x_i|$ is the length of the longest sequence, and $k = \max_j \max_i |x_{i,j}|$ is the length of the longest string in any of the sequences.
  Then there is an algorithm solving DSSSP on $X$ (for any $\lambda$) that runs in time
  $\smash{O(n^{m^2} \cdot 2^m \cdot k^m \cdot \mathrm{poly}(k,m,n))}$.
\end{theorem}
\begin{proof}
  The approach (see Figure~\ref{fig:brute-force-algorithm} in Appendix~\ref{sect:proofs})
  is brute force:  we look at every possible tree topology $\tau$, which specifies, for any $x,x' \in X$, the indices $i$ and $i'$ into $x$ and $x'$ at which they diverge. 
  Every $\tau$ defines a set of nonbranching segments between divergences; the summary tree problem is then a collection of summary ``path'' problems, one per segment.

  The path problem can be seen as multiple sequence alignment (MSA)~\cite{sankoff1975minimal}, solvable via dynamic programming~\cite{carrillo1988multiple} (see also~\cite{waterman1976some,kececioglu1993maximum,raphael2004novel}).  To implement the MSA dynamic program, we must compute the cost of assigning a set of strings $S = \set{s_1, s_2, \ldots, s_\ell}$ to a single node $u$.  If we label $u$ with the string $z$, then the alignment cost of this node is $\lambda + \sum_i \ed(s_i,z)$; thus the best label is the \emph{string median} of $S$---that is, the string $z$ minimizing the summed edit distance to the strings in $S$.  While string median is NP-hard~\cite{de2000topology}, a dynamic programming algorithm solves string median for a fixed number of strings~\cite{sankoff1975minimal} (see also \cite{sankoff1976frequency,kruskal1983overview,nicolas2005hardness}).
    
  There are $\smash{O(n^{m \cdot (m-1)})}$ tree topologies.  Each defines a tree with $\mathop{\le} m$ leaves and thus $\mathop{\le} 2m$ nonbranching segments.
 
  Each segment contains $\mathop{\le} m$ subsequences, each of length $\mathop{\le} n$; thus each multiple sequence alignment requires $O(\smash{n^m})$ time~\cite{carrillo1988multiple}.  Whenever we compute the string median of a candidate node, we have $\mathop{\le} m$ strings each of length $\mathop{\le} k$; computing these medians takes $O(\smash{(2k)^m})$ time~\cite{sankoff1975minimal}. Finally, it takes $\mathrm{poly}(k,m,n)$ time to compute $\error_\lambda(T)$ for each tree. Thus the overall running time is $\smash{O(n^{m\cdot(m-1)} \cdot n^m \cdot (2k)^m \cdot \mathrm{poly}(n,m,k))}$.
\end{proof}


\section{An Efficient Heuristic for Larger Sets of Sequences}
\label{sec:tree_build}

Given DSSSP's hardness (Theorem~\ref{thm:hardness-of-reconstruction}) and the abominable running time of our exact algorithm (Theorem~\ref{thm:dsssp_is_fixed_parameter_tractable}), we turn here to an efficient heuristic for DSSSP with larger $m$.  Our algorithm is greedy, and seeks to repeatedly identify the pair of sequences in $X$ with the longest shared prefix, and then merge that shared prefix into a single sequence (as in \buildbifurcation).  See Figure~\ref{fig:tree_reconstruction}.  There are several issues that we must resolve:

\begin{figure}[tb]
    \centering
    \includegraphics[width=0.8\textwidth]{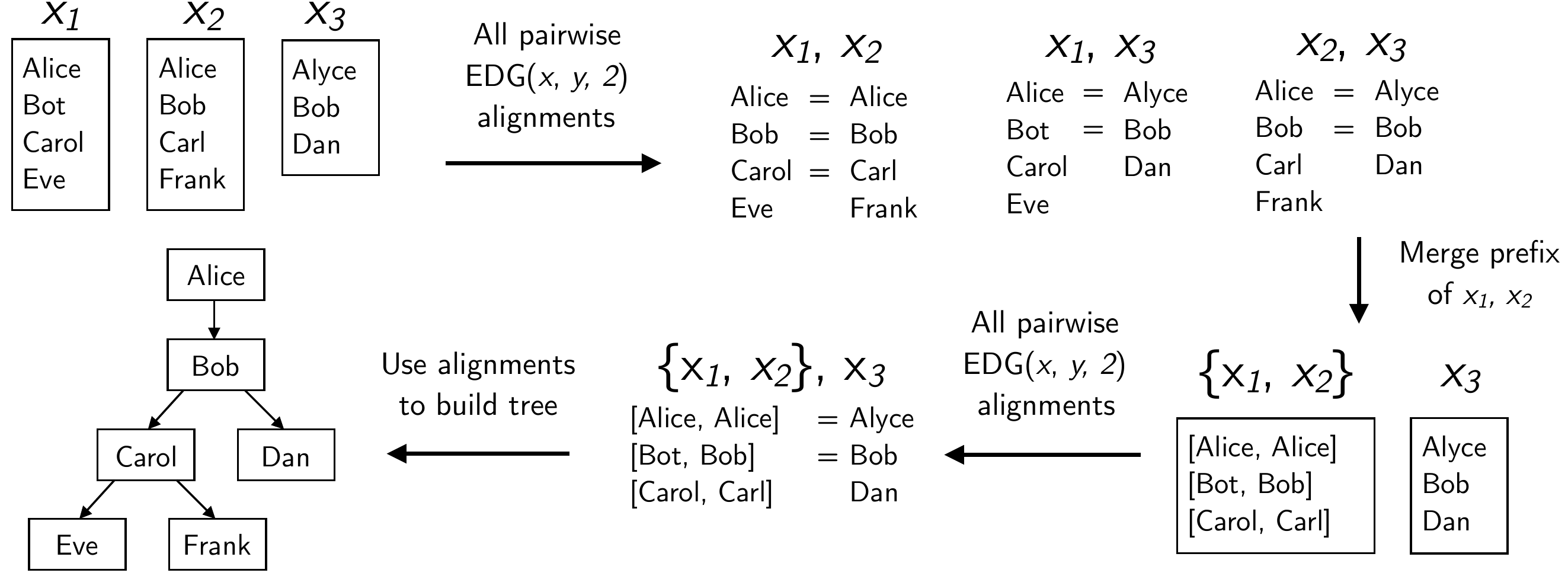}
    \caption{An example run of $\buildtree(X,\lambda)$ for $\lambda=2$ and the sequences from Figure~\ref{fig:tree_and_petitions}.}
    \label{fig:tree_reconstruction}
\end{figure}

\begin{description}
\item[Measuring and merging the shared prefix of $x_i$ and $x_j$.]  To calculate how well $x_i$ and $x_j$ match, we compute the $\edg(x_i,x_j,\lambda)$ alignment.  Define the \emph{number of substitutions \emph{(ignoring insertions and deletions)} in the pre-divergence section $p_{i,j}$ of this alignment} as their overlap.  Then, for the pair $\smash{\set{x_i,x_j}}$ with the largest overlap, replace $\smash{\set{x_i,x_j}}$ with $p_{i,j}$ in $X$. Note that $p_{i,j}$ is a sequence of \emph{lists of strings,} not a sequence of \emph{strings;} thus we need to generalize $\edg$ to sequences of lists of strings, not just individual strings.
  
\item[Reconciling labels in the final resulting tree.] Repeating this merging process will define a tree, except that each node is labeled by a \emph{list} of strings, not just one.

To produce the final labels,
we use the \emph{medoid} string.  For a list of strings $A$, the medoid of $A$ is the string in $A$ whose sum of edit distances to strings in $A$ is minimized.\footnote{When we say ``the'' medoid of $A$, we mean the lexicographically first medoid of $A$. (Which medoid we choose never affects the sum of the distances at hand---e.g., in the sums in~\eqref{eq:cost}---but we need to identify one in particular for the summands to be well-defined.)  Note that the medoid, unlike the median, must be an element of $A$; we use it because it is efficiently computable (unlike median) and is a $(2-o(1))$-approximation to the median (an implication of the triangle inequality).} 
\item[Generalizing $\edg$ to lists of strings.] 

  Define the edit distance between a string $x$ and a set of strings $X$ as $\ed(x, X) := \ed(x,\medoid(X))$---using the medoid of $X$ as its representative string.
Then, letting $\cost(A,B)$ denote the cost of merging lists $A$ and $B$, we define
\begin{equation}\label{eq:cost}
    \cost(A, B) := \textstyle \sum\limits_{x \in A \cup B} \ed(x, A \cup B) - \sum\limits_{x \in A} \ed(x,A) - \sum\limits_{y \in B} \ed(y,B).
\end{equation}
This cost quantifies the amount of additional disagreement incurred by merging the lists, relative to leaving them separate. Insertion and deletion costs are found using (\ref{eq:cost}) with a list of empty strings of appropriate length in place of $x_i$ or $y_j$. We thus define the $\edg$ recurrence (cf.~Equation~(\ref{eq:edg-recurrence})) for sequences of lists of strings $x$ and $y$ as follows:
\[
   \edg(i,j,\lambda) = \text{min}  \nonumber
	\left\{\begin{array}{@{}l@{~+~}lll}
    	\edg(i+1, j+1, \lambda) & \lambda + \cost(x_i, y_j) &\text{(substitution)}\\
                 \edg(i, j+1, \lambda) &\lambda +
      \cost(\set{\text{$|x_i|$ copies of $\eps$}}, y_j)&\text{(insertion)}\\
    	\edg(i+1, j, \lambda) & \lambda + \cost(x_i, \set{\text{$|y_j|$ copies of $\eps$}}) &\text{(deletion)}\\
        \lambda (|x|-i + 1)& \lambda(|y|-j + 1)&\text{(give up)}
        \end{array}\right.
\]
\end{description}

Define $\buildtree(X,\lambda)$ as the greedy iterative algorithm suggested above:  until there is only one sequence left in $X$, find the pair of sequences $x_i, x_j \in X$ with the largest number of substitutions in $\edg(x_i, x_j, \lambda)$, and replace $\set{x_i, x_j}$ by their merged prefix $p_{i,j}$.  (Save the post-divergence branches of the bifurcation; we will reattach those branches at the bottom of $p_{i,j}$ in the final tree.)  When there is only one sequence left, reattach all of the saved branches, and replace each node's list-of-strings label by the medoid of that label list.
(See Figure~\ref{fig:tree_reconstruction}.)


\section{Evaluation and Parameter Selection}
\label{sect:evaluation}

\buildtree is suboptimal both because greedy merging can yield a poor topology and because medoids can be poor node labels; see Examples~\ref{example:bad-example-of-greedy-merging} and~\ref{example:bad-example-of-medoid}.
Still, we will show that it nonetheless performs well on simulated data, suggesting that it is a good heuristic.  

\begin{example}[A bad example for greedy merging]
  \label{example:bad-example-of-greedy-merging}
  Consider the instance
  \[
    x_1 = \text{\tt a b c}
    \qquad\qquad
    x_2 = \text{\tt a b d}
    \qquad\qquad
    x_3 = \text{\tt a e d}
    \qquad\qquad
    x_4 = \text{\tt a e f}
  \]
  with $0.5 < \lambda < 1$. The optimal tree $T^*$ is shown below, with $\error_\lambda(T^*) = 7\lambda$. However, \textsc{BuildTree} will choose to merge sequences with the most closely aligned pair of sequences according to the {\edg} alignment. In this case, it would choose to merge $x_2$ and $x_3$ first. The tree $T$ returned by \buildtree has $\error_\lambda(T) = 5\lambda + 2$, making it suboptimal.
\begin{center}\hfill
\begin{tikzpicture}[
        every node/.style = {align=center},
        level 1/.style={sibling distance=7em},
        level 2/.style={sibling distance=4em},
        level distance=0.75cm]
    \node {{\tt a}}
        child { node {{\tt b}} 
            child { node {{\tt c}} }
            child { node {{\tt d}} }
        } 
        child { node {{\tt e}}
            child { node {{\tt d}} }
            child { node {{\tt f}} }
        }
    ;
    \node at (-2, 0) {$T^*$};
\end{tikzpicture} \hfill
\begin{tikzpicture}[
        every node/.style = {align=center},
        level 1/.style={sibling distance=7em},
        level 2/.style={sibling distance=4em},
        level distance=0.75cm]
    \node {{\tt a}}
        child { node {{\tt b}} 
            child { node {{\tt c}} }
            child { node {{\tt d}} }
            child { node {{\tt f}} }
        } 
    ;
    \node at (-2, 0) {$T$};
\end{tikzpicture} \hfill\text{}
\end{center}
\end{example}

\begin{example}[A bad example for medoids]
  \label{example:bad-example-of-medoid}
  Consider the set $S = \set{\mathtt{XABC},\mathtt{AXBC},\mathtt{ABXC},\mathtt{ABCX}}$, where each string sequence contains just one string.  For large $\lambda$, the optimal tree is just a single node.  Here the optimal label is the median $\mathtt{ABC}$, which has total edit distance 4 to the set $S$; the medoid label $\mathtt{ABCX}$ has total edit distance 6 to $S$.
\end{example}

\subsection{Generating Synthetic Data}
\label{sect:generating-synthetic-data}

We will generate synthetic data based on several parameters:  the number $m$ of string sequences, the string length $k$, a string substitution probability $\sigma_s$, a string deletion probability $\delta_s$, a character substitution probability $\sigma_c$, and a character deletion probability $\delta_c$.

We generate trees using branching processes (i.e., \emph{Galton--Watson trees}~\cite{watson1875probability}): each node chooses to have exactly $i$ children with probability $p_i$; we fix ${p_0=0.03}$, ${p_1=0.94}$, and ${p_2=0.03}$ to approximate real email petition data~\cite{LNK08,GJ10}.  We generate a Galton--Watson tree $T$ until it has $m$ leaves, restarting if the branching process terminates early.  We label each node $u \in T$ with a random alphabetic string $\ell(u)$ of length $k$.  Let $x_i$ denote $\L{T}{u_i}$ for the $i$th leaf $u_i$.  Call $T$ the \emph{true tree,} $\ell$ the \emph{true labels,} and $x_i$ the \emph{true sequences.}

Now, we simulate noisy propagation.  To mirror petition data derived from low-quality scans of printed emails, we introduce string- and character-level errors in separate phases:
\begin{enumerate}
\item[(1)] \emph{string-level errors (which are inherited).}
  Each node $u$ inherits from its parent $p$ the noisy history $h_p$ of its ancestral labels.  (The root ``inherits'' an empty sequence.)  The node $u$ (further) corrupts $h_p$:  for each string in $h_p$, substitute it with a random alphabetic string of length $k$ with probability $\sigma_s$, and delete it with probability $\delta_s$.  Finally, node $u$ appends its true label $\ell(u)$ to $h_p$; call the resulting sequence $h_u$.
  Now each leaf $u$ stores $h_u$, a noisy version of $\L{T}{u}$.  Let $X' = \set{x_1', \dots, x_m'}$ be the set of histories at the leaves.

\item[(2)] \emph{character-level errors (which appear independently).} For each $x_i' \in X'$, substitute each character in each string in $x_i'$ with a random character with probability $\sigma_c$, and delete the character with probability $\delta_c$.  Let $X'' = \set{x_1'', \dots, x_m''}$ be the resulting sequences.
\end{enumerate}
Our experiments use string length $k=25$, string error rates $\sigma_s = \delta_s = 0.001$, character error rates $\sigma_c = \delta_c = 0.1$, and $m \in \set{15, 100}$.  Because string-level errors compound, $0.001$ is a nontrivial error rate (roughly comparable to the $10\%$ character-level error rate).

\subsection{Comparing Reconstruction to Synthetic Ground Truth}
\label{sect:comparing-reconstruction-to-ground-truth}

We generate a true tree $T$, with true sequences $X = \set{x_1, \ldots, x_m}$ and corrupted sequences $X'' = \set{x_1'', \dots, x_m''}$.  We then use our heuristic algorithm to build a reconstructed tree $T' = \buildtree(X'',\alpha)$, for some choice of a node-cost parameter $\alpha$ to use in the reconstruction algorithm. (See Section~\ref{sec:parameter_selection} regarding how to choose $\alpha$.)
Note the distinction between the reconstruction parameter $\alpha$ and the evaluation parameter $\lambda$: $\buildtree$ seeks to optimize $\error_\alpha(T')$ for some value of $\alpha$, but we can assess the quality of $T'$ using $\error_\lambda(T')$ whether or not $\alpha = \lambda$, to evaluate sensitivity to $\alpha$.

We will compare the quality of \buildtree to the threshold-based reconstruction algorithm from~\cite{LNK08}, which we briefly describe here.  (1) Construct a weighted directed graph $G$ with nodes labeled by all strings in all $x \in X$, and with an $a \to b$ edge if string $a$ immediately precedes string $b$ in any sequence $x \in X$.  The weight of this edge is the \emph{number} of sequences $x \in X$ containing $a$ and $b$ successively.  (2) To handle minor signature errors, treat two signatures as equivalent if they follow  equivalent signatories and have an edit distance below a fixed threshold $\beta$. (3) Compute the tree as a max-weight spanning arborescence of $G$, with extraneous nodes pruned away.  (Note that $\beta$ is on the same scale as $\alpha$; it defines a cutoff of $\ed(a,a')$ indicating when $a$ and $a'$ should be assigned to two nodes versus one.)

\subparagraph*{Measures of performance.}
We assess the quality of our reconstructed trees $T'$ in two ways.

First, we use the $\error_\lambda(T')$ measure from~\eqref{eq:error}, the sum of $\lambda \cdot |T'|$ and all label distances. 
Second, we use \emph{tree edit distance} (TED) to compare the structure of $T'$ to the true tree~$T$.  Edit distance between unordered trees is NP-hard~\cite{zhang1992editing}, so we estimate the distance by ordering our trees and computing \emph{ordered} TED using Zhang--Shasha~\cite{ZS89}, as implemented by Henderson~\cite{zssmodule}.  (Specifically, we recursively order sibling nodes in $T'$ to minimize the number of leaf order inversions with respect to $T$.  As the number of siblings is typically small,
this order is not expensive to compute.) To better interpret our results, we also compute TED to a 
tree
with the correct structure but with label corruptions, denoted $\tilde{T}$; this represents the best reconstruction we could hope for (without computing medians). We construct $\tilde{T}$ by labeling the ancestors of each leaf $u_i$ using the noisy history $h_{u_i}$ (picking the medoid for nodes assigned multiple labels), deleting nodes with empty labels.%
\begin{figure}[t]
    \centering
    \begin{subfigure}[t]{0.95\textwidth}
        \centering
        \includegraphics[width=\textwidth]{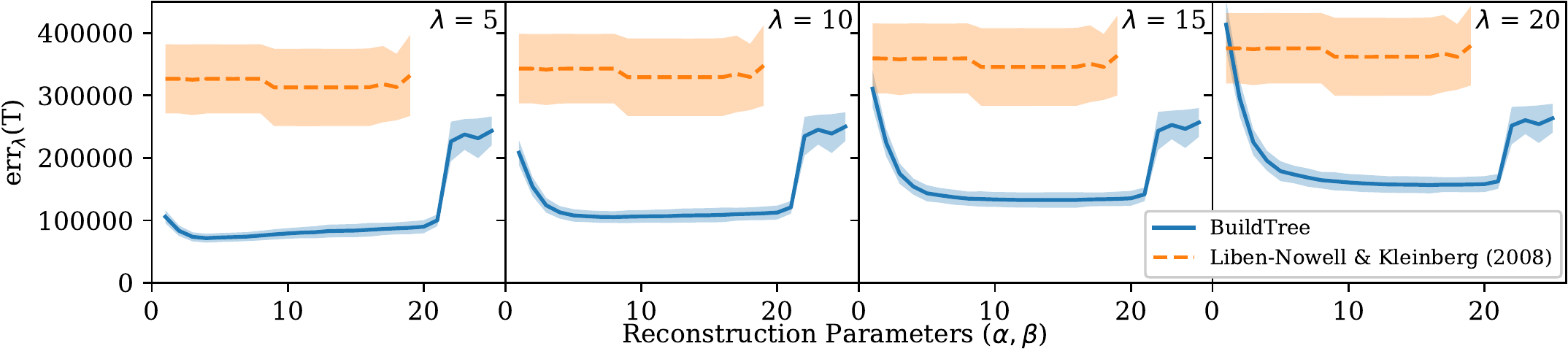}
        \caption{Error under $\error_{\lambda}(T)$ with $|X| = 100$, for several $\lambda$ values, averaged across 8 trials.}
        \label{fig:reconstruction_quality:error_measure:one-big-tree}
    \end{subfigure}

    \smallskip
    \begin{subfigure}[t]{0.48\textwidth}
        \centering
        \includegraphics[width=\textwidth]{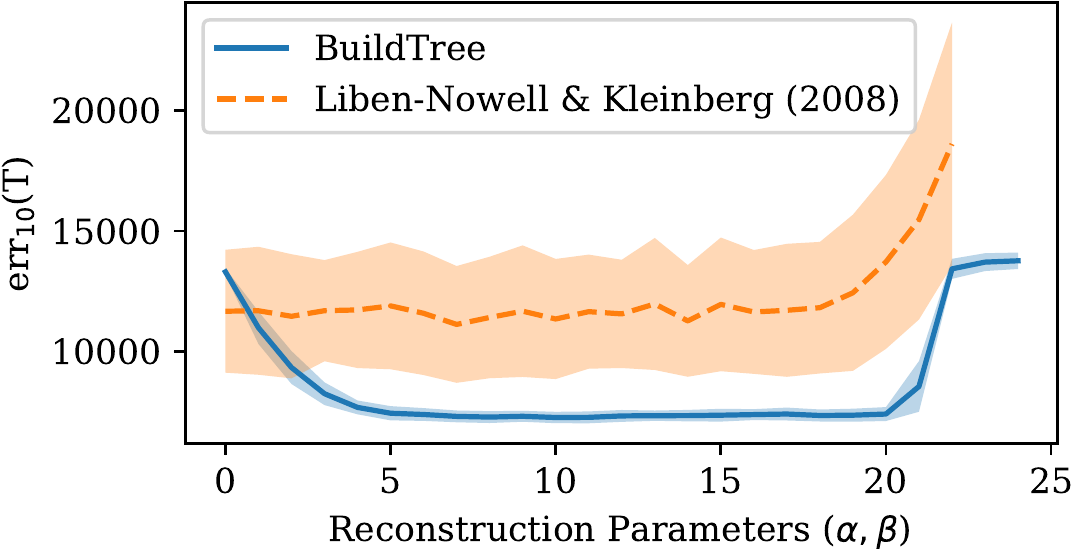}
        \caption{Error under $\error_{10}(T)$ with $|X| = 15$, averaged across 500 trials.  Smaller $X$ makes the error smaller than in (a), but the trend is similar.}
        \label{fig:reconstruction_quality:error_measure:many-small-trees}        
    \end{subfigure}%
    \quad 
    \begin{subfigure}[t]{0.48\textwidth}
        \centering
        \includegraphics[width=\textwidth]{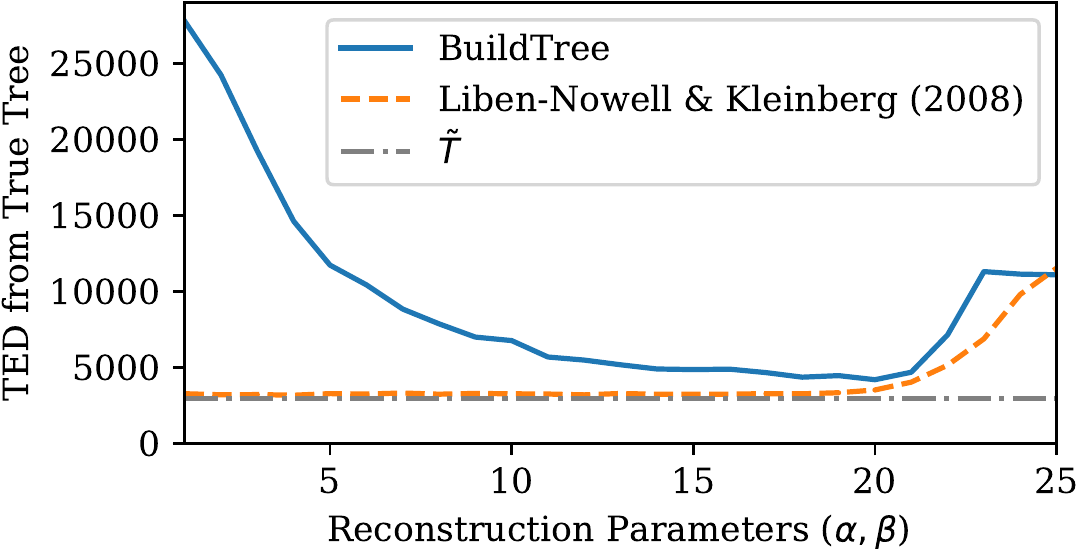}
        \caption{Error under ordered tree edit distance (TED) on a $|X| = 15$ dataset. The dash-dotted line shows TED between $\tilde T$ and the real tree.}
        \label{fig:reconstruction_quality:ZS_measure:many-small-trees}        
    \end{subfigure}
    \caption{Comparing $\buildtree(X,\alpha)$ to the algorithm from~\cite{LNK08} with edit-distance threshold $\beta$.  Shaded regions in (a) and (b) show standard deviations across the stated number of trials.  Observe the flat-bottomed basin shape of the error curve for $\buildtree(X,\alpha)$ as $\alpha$ varies:  high error rates when $\alpha < 5$ and when $\alpha > 20$, and roughly constant low error for all $\alpha$ in between.}
    \label{fig:reconstruction_quality}
\end{figure}%

\subparagraph*{Evaluation results.}
We generated a tree $T$ with $m = 100$ sequences, and corrupted sequences $X''$ (with 8 independent trials generating different $X''$ from $T$). Figure~\ref{fig:reconstruction_quality:error_measure:one-big-tree} compares  $\error_\lambda(\cdot)$ of our reconstruction 
with the method from~\cite{LNK08}, showing that \buildtree performs significantly better over a range of $\lambda$ values.  These trees were too big to compute tree edit distance, which is computationally prohibitive. 

In addition, we generated a tree with $m = 15$ sequences and introduced error in 500 trials. Figure~\ref{fig:reconstruction_quality} shows both $\error_{10}(T)$ (Figure~\ref{fig:reconstruction_quality:error_measure:many-small-trees}) and TED (Figure~\ref{fig:reconstruction_quality:ZS_measure:many-small-trees}) between the reconstructed tree $\widehat{T}$ and the true tree $T$. For many values of $\alpha$, $\buildtree(X'',\alpha)$ produces significantly better solutions than the method from~\cite{LNK08}, as shown by the $\error_{10}(T)$ measure. It is also competitive at $\alpha=20$ according to (ordered) tree edit distance, a metric \buildtree was not designed to optimize.

\subsection{How Should the Node Cost be Selected?}
\label{sec:parameter_selection}

To reconstruct a tree from real sequence data, we must select a value for the reconstruction parameter $\alpha$.  If $\alpha$ is too low, nodes are too cheap and trees diverge too early; if $\alpha$ is too high, nodes are too costly and trees branch too rarely.  But what value should we choose?

Intuitively, we wish to map two (corrupted) strings to the same node if they are (corruptions of) the same true string.  Imagine two strings $u$ and $v \neq u$, and let $u'$, $u''$, and $v'$ be the result of independently introducing errors to $u$, $u$, and $v$, respectively.  We desire a value of $\alpha$ so that $u'$ and $u''$ would probably be mapped to the same node, but $u'$ and $v'$ probably diverge.  That is, the cost of creating a new node should be greater than the cost of aligning two corrupted versions of the same string, but the cost of diverging should be less than the cost of deleting/inserting all remaining strings.  Thus we want $E[\ed(u', u'')] < \alpha < E[\ed(u', v')]$.  

$E[\ed(u', u'')]$ depends on the error rate of the corruption process and must be estimated from data.  But estimates of $E[\ed(u', v')]$ appear in the literature if the true strings $u$ and $v$ are uniformly random and have equal length.  If the cost of substitution is twice that of insertion/deletion,
then we can calculate $E[\ed(u', v')]$ using Chv\'{a}tal--Sankoff numbers~\cite{chvatal_sankoff_1975} (the expected longest common subsequence length for two equal-length random strings); for unit edit costs, it is conjectured that $E[{\ed}(u', v')] \approx |u'|(1-\frac{1}{|\Sigma|})$ for random equal-length strings over $\Sigma$~\cite{Nav01}. If strings are not uniformly random or have different lengths, then we can estimate $E[{\ed}(u', v')]$ from data and subsequently select an $\alpha$ between the upper and lower bounds.
(In fact, there is some evidence that many values of $\alpha$ in this range perform well; see the flat-bottomed basin shape of Figure~\ref{fig:reconstruction_quality:error_measure:many-small-trees}.)

\section{Discussion and Future Work}

We described an efficient, practical heuristic to reconstruct a propagation tree from a noisy set of diverging string sequences---and in Section~\ref{sect:evaluation} we showed that \buildtree performs well on synthetic data.  But, after all, the motivation for introducing this particular theoretical problem was for its application to real data, particularly chain-letter petitions.  Rigorously testing \buildtree on real data, then, is perhaps the most natural direction for future work.
(Testing on more realistic synthetic data is also an interesting future direction.  Our data-generation process in Section 7.1 is unrealistic in a number of ways, perhaps most strikingly in its assumption that a name is a length-$k$ alphabetic string chosen uniformly at random.  More realistic randomized name-generation processes would make the synthetic task more similar to the real one.)

That said, there are several potentially interesting theoretical avenues for further exploration of DSSSP, too.  The problem is (at least theoretically) tractable for any fixed number $m$ of string sequences---but the dependence on $m$ in our brute-force algorithm is brutal.  Is there a more efficient algorithm for small $m$?  Or are there efficient algorithms with provable approximation guarantees for general $m$?  We can approximate the best labels for a fixed tree topology using medoids or, even better, using a PTAS for the string median problem~\cite{li2002finding}.  Identifying the best tree topology seems more challenging, but perhaps this topological source of error in \buildtree (or some other heuristic) can be bounded.  

There is another set of interesting open questions related to efficient algorithms for DSSSP when $\lambda$ is small.
(At the other extreme, the problem remains intractable when $\lambda$ is very large:  even if the optimal tree's topology is the trivial non-branching one, as in Figure~\ref{fig:ex:giant-lambda}, choosing the labels requires repeatedly solving instances of the NP-hard string median problem.)

DSSSP is trivial when $\lambda = 0$; the diverge-at-the-root tree (as in Figure~\ref{fig:ex:small-lambda}) is optimal, with total cost $0$.  Even for strictly positive but small values of $\lambda$, there is an easy solution:  if $\lambda \le \smash{\frac{1}{nm}}$, it optimal to diverge upon encountering even a one-character difference between strings (i.e., the optimal summary tree is precisely the trie representation of the set $X$).  
Do related approaches make the problem tractable for bigger values of $\lambda$---e.g., if $\lambda$ is small enough that we can only afford a bounded budget of edits in node labels? For how large a value of $\lambda$ are there exact polynomial-time algorithms?

\newpage
\bibliography{paper}

\newpage
\appendix

\section{Proofs Omitted from the Main Paper}
\label{sect:proofs}

\subsection{Proofs from Section~\protect\ref{sect:the-problem}}

\bbisoptimal*
\begin{proof}
  We show that $T^\ast := \buildbifurcation(x, y, \lambda)$ satisfies $\error_\lambda(T^\ast) \leq \edg(x, y, \lambda)$ (Lemma~\ref{lemma:build}), and that $\error_\lambda(T) \geq \edg(x, y, \lambda)$ for any valid summary tree $T$ (Lemma~\ref{lemma:opt}).  Thus $T^\ast$ is optimal.
\end{proof}

\begin{lemma} \label{lemma:build}
The tree $T^*$ built by $\buildbifurcation(x, y, \lambda)$ has $\error_\lambda(T^*) \leq \edg(x, y, \lambda)$.
\end{lemma}
\begin{proof}
By induction on the lengths of $x$ and $y$. 

\emph{Base Case:}
If at least one of $x$ and $y$ are empty, then suppose without loss of generality that $x$ is empty. Then the optimal $\edg$ alignment of $x$ and $y$ diverges immediately, with a cost of $\edg(x, y, \lambda) = \lambda |y|$. Thus, {\buildbifurcation} returns the path $T^\ast$ labeled by $y$, which has $\error_\lambda(T^*) = \lambda|y|$. Therefore $\error_\lambda(T^*) \leq \edg(x, y)$.

\emph{Inductive Case:}
Assume $x$ and $y$ are both non-empty sequences. Consider the possible cases for the $\edg$ alignment:
\begin{enumerate}[a.]
   \item \emph{$\edg$ chooses to align $x_1$ and $y_1$ with each other.} Then $\edg(1,1, \lambda) = \edg(2,2, \lambda) + \lambda + \ed(x_1, y_1)$. The tree produced in this case has its top non-sentinel node labeled by $x_1$ followed by the subtree $\buildbifurcation(x_{2 \ldots |x|}, y_{2 \ldots |y|}, \lambda)$. By inductive hypothesis, the error of the subtree is at most $\edg(2,2, \lambda)$. One possible {\ad} alignment of $x$ to $\L{T^*}{v_x}$ substitutes $x_1$ for $[\L{T^*}{v_x}]_1 = x_1$, and a possible alignment of $y$ to $\L{T^*}{v_y}$ substitutes $y_1$ for $[\L{T^*}{v_y}]_1 = x_1$. Using these (potentially suboptimal) alignments, the error of the tree would be at most $ \edg(2,2,\lambda) + \lambda + \ed(x_1, y_1)$, which is exactly the cost of $\edg(x, y, \lambda)$. The optimal {\ad} alignment must have cost less than or equal to this particular alignment, so $\error_\lambda(T^*) \leq \edg(x, y, \lambda)$.

   \item \emph{$\edg$ chooses to align $y_1$ with $\eps$.}
    Then $\edg(1, 1, \lambda) = \edg(1, 2, \lambda) + \lambda + \ed(\varepsilon, y_1)$. The tree produced in this case has its top non-sentinel node labeled by $y_1$ with the subtree produced by $\buildbifurcation(x_{1 \ldots |x|}, y_{2\ldots |y|}, \lambda)$. By the inductive hypothesis, the error of the subtree is at most $\edg(1, 2, \lambda)$. One possible {\ad} alignment of $x$ to $\L{T^*}{v_x}$ inserts $[\L{T^*}{v_x}]_1 = y_1$ and then optimally aligns $x$ to the rest of $\L{T^*}{v_x}$. Similarly, one alignment of $y$ to $\L{T^*}{v_y}$ substitutes $y_1$ for $[\L{T^*}{v_y}]_1 = y_1$ and optimally aligns the rest. Therefore, the total error of the tree produced by $\edg(1,1, \lambda)$ with this alignment is at most $\edg(1, 2, \lambda) + \lambda + \ed(\varepsilon, y_1)$, which is indeed the cost of $\edg(x, y, \lambda)$. The optimal {\ad} alignment must have cost less than or equal to this alignment, so $\error_\lambda(T^*) \leq \edg(x, y, \lambda)$. \label{case:A1:b}

   \item \emph{$\edg$ chooses to align $x_1$ with $\eps$.}
     This case is symmetric to Case (\ref{case:A1:b}). 

   \item \emph{$\edg$ chooses to diverge immediately.} Then $\edg(x,y, \lambda) = \lambda|x| + \lambda|y|$. \buildbifurcation constructs $T^*$ by labeling two paths by $x$ and $y$ and placing them under a sentinel root. This tree $T^*$ has $\error_\lambda(T^*) =  \lambda|x| + \lambda|y| \le \edg(x, y, \lambda)$.   \qedhere
\end{enumerate}
\end{proof}

\begin{lemma}\label{lemma:opt}
$\edg(x, y, \lambda) \leq \error_\lambda(T)$ for any valid tree $T$ representing $x$ and $y$.
\end{lemma}
\begin{proof}
We proceed by induction on the lengths of $x$ and $y$.

\textit{Base Case:} If at least one of $x$ and $y$ are empty sequences, then suppose without loss of generality that $x$ is empty. The optimal {\edg} alignment gives up immediately, for a cost of $\lambda |y|$. Any tree $T$ representing $x$ and $y$ must have at least $|y|$ nodes for $\ad(y, \L{T}{v_y})$ to be defined, so $\error_\lambda(T)\ge \lambda |y|$. Thus $\edg(x, y, \lambda) \leq \error_\lambda(T)$.

\textit{Inductive Case: } Assume $x$ and $y$ are both non-empty sequences and let $T$ be any valid tree representing $x$ and $y$. If the root of $T$ has more than two children, we can strictly decrease $\error_\lambda(T)$ by deleting the children (and their subtrees) that do not contribute to $\L{T}{v_x}$ or $\L{T}{v_y}$, bringing the number of children down to at most two. We can therefore ignore this case. If the (sentinel) root of $T$ has exactly two children, then $T$ must have at least $|x| + |y|$ non-sentinel nodes to be valid, so $\error_\lambda(T) \ge \lambda|x| + \lambda|y|$. One possible $\edg$ alignment gives up immediately for a cost of $\lambda|x| + \lambda|y|$, so $\edg(x, y, \lambda) \le \lambda|x| + \lambda|y| \le \error_\lambda(T)$.

Now suppose the sentinel root of $T$ has one child labeled $\ell$. Call the subtree under this node $T'$ (adding a sentinel root to $T'$ to ensure it is a tree). Consider the cases for the optimal $\ad$ alignments of $x$ and $y$ to their paths in $T$, with respect to $\ell$.
\begin{enumerate}[a.]
        \item If $x_1$ and $y_1$ are both substituted for $\ell$ in the optimal alignments, then $\error_\lambda(T) = \error_\lambda(T') + \lambda + \ed(x_1, \ell) + \ed(y_1, \ell)$. Notice that $T'$ is a summary tree of $x_{2\ldots|x|}$ and $y_{2\ldots|y|}$. By the inductive hypothesis, $\edg(2, 2, \lambda) \le \error_\lambda(T')$. One possible {\edg} alignment substitutes $x_1$ for $y_1$, with a cost of $\edg(2, 2, \lambda) + \lambda + \ed(x_1, y_1)\le \error_\lambda(T') + \lambda + \ed(x_1, \ell) + \ed(y_1, \ell)$ (using the triangle inequality for edit distance). Thus $\edg(x, y, \lambda) \le \error_\lambda(T)$. 
        
        \item If $x_1$ is substituted for $\ell$ and $\ell$ is inserted against $y$ in the optimal alignment, then $\error_\lambda(T) = \error_\lambda(T') + \lambda + \ed(x_1, \ell) + \ed(\ell, \varepsilon)$. Notice that $T'$ is a summary tree of $x_{2\ldots|x|}$ and $y_{1\ldots|y|}$. By the inductive hypothesis, $\edg(2, 1, \lambda) \le \error_\lambda(T')$. One possible {\edg} alignment deletes $x_1$, with a cost of $\edg(2, 1, \lambda) + \lambda + \ed(x_1, \varepsilon)\le \error_\lambda(T') + \lambda + \ed(x_1, \ell) + \ed(\ell, \varepsilon)$ (using the triangle inequality for edit distance). Thus $\edg(x, y, \lambda) \le \error_\lambda(T)$. \label{case:A2:b}
        
        \item The case where $y_1$ is substituted for $\ell$ and $\ell$ is inserted against $x$ is symmetric to Case~(\ref{case:A2:b}).
        
        \item If $\ell$ is inserted against both $x$ and $y$, then $T'$ is a summary tree of $x$ and $y$ with smaller error. We can thus reduce this case to one of the previous cases.  \qedhere
\end{enumerate}
\end{proof}

\subsection{Proofs from Section~\protect\ref{sect:larger-sets}}

As suggested in Section~\protect\ref{sect:larger-sets}, a reduction from the String Median problem shows that DSSSP is hard for sufficiently large $\lambda$---large enough that the optimal summary tree never branches---even if all $m$ string sequences have length~$1$.
While the reduction from Median String is simpler, we also provide an alternative reduction from Shortest Common Supersequence (SCS)~\cite{RAIHA1981187}, which applies for smaller $\lambda$ as well. (Note that SCS is hard in general, but for a small number of strings it is efficiently solvable~\cite{fraser1995subsequences}.)

\dssspishard*
\begin{proof}
   Given an instance of SCS consisting of a set of strings $\set{x_1, \dots, x_m}$, construct the following instance of DSSSP. Let $\lambda = 1$ and treat each string $x_i$ as a sequence of single-character strings. Additionally, set the $\ed$ cost of non-identical substitutions and deletions to $\infty$ and the cost of identical substitutions and insertions to $0$.

   Let $s$ be the length of the SCS of $\set{x_1, \dots, x_m}$. In this DSSSP instance, one possible summary tree $T^*$ consists of a single length-$s$ path labeled by the SCS of $\set{x_1, \dots, x_m}$.

   To show that $T^*$ is optimal, we show that any finite-error summary tree $T'$ can be converted into a path labeled by a supersequence of $\set{x_1, \dots, x_m}$ without increasing its error. Note that the $\ad$ alignments in $T'$ must have only insertions and substitutions of identical symbols if $\error_\lambda(T')$ is finite. To convert $T'$ into a path, pick an arbitrary leaf $v$ of $T'$. If $T'$ is not a path already, then it has some branch that diverges from the root-to-$v$ path. Take the branch after this divergence and move its root to be a child of $v$. $T'$ now has one fewer leaf, but the same number of nodes. Moreover, the $\ad$ alignments still have cost zero, as we can align strings to their modified paths in $T'$ by using more zero-cost insertions. By repeating this process until only one leaf remains, we can convert $T'$ into a path with the same $\error_\lambda(T')$, which must then be labeled by a supersequence of $\set{x_1, \dots, x_m}$. Therefore $T^*$ is an optimal DSSSP solution, as it is the cheapest of all such paths.

   If we could solve DSSSP in polynomial time, then we could efficiently convert this solution into a path labeled by a shortest common supersequence of $\set{x_1, \dots, x_m}$ using the iterative method described above.
 \end{proof}

In Theorem~\ref{thm:dsssp_is_fixed_parameter_tractable}, we describe an algorithm for DSSSP.  The algorithm's running time is $\smash{O(n^{m^2} \cdot 2^m \cdot k^m \cdot \mathrm{poly}(k,m,n))}$, where $m$ is the number of sequences, $n$ is the length of the longest sequence, and $k$ is the length of the longest string in any sequence.  Figure~\ref{fig:brute-force-algorithm} shows the pseudocode for the algorithm described in the proof of Theorem~\ref{thm:dsssp_is_fixed_parameter_tractable}.

\begin{figure}[h!]
  \fbox{\begin{minipage}{\linewidth}
    \begin{tabbing}
      xx \= xx \= xx \= xx \kill
      For each topology $\tau$, specified by $m \cdot (m-1)$ numbers $\ell_{i,j} \in \set{0, 1, \ldots, n}$:\\
      \> \comment{$\ell_{i,j}$ specifies the index of the last string in $x_i$ before the divergence point of $x_i$ and $x_j$} \\
      \> If $\tau$ does not define a tree (i.e., the divergence points are inconsistent), skip this $\tau$.\\
      \> Otherwise, for each nonbranching segment $\sigma$ of $\tau$:\\
      \> \> (1) Identify those sequences in $X$ with subsequences present in $\sigma$,\\
      \> \> \> and let $Y_{\sigma}$ denote the corresponding subsequences of the sequences in $X$.\\
      \> \> (2) Compute the multiple sequence alignment of $Y_{\sigma}$ via dynamic programming.\\
      \> \> \> The cost of associating strings $S = \set{s_1, s_2, \ldots, s_t}$ at the same node in the tree\\
      \> \> \> is $\lambda + \sum_i \ed(s_i,\mathsf{median}(S))$.
      Compute $\mathsf{median}(S)$ via dynamic programming.\\
      \> \> (3) Assemble the alignments for each segment into a tree $T_\tau$.\\
      Return the $T_\tau$ of minimum cost, over all values of $\tau$.
    \end{tabbing}
  \end{minipage}}
\caption{The brute-force algorithm for DSSSP.  Classical Multiple Sequence Alignment algorithms are used to align \emph{sequences of letters} (a.k.a.~strings) rather than, as in our case, \emph{sequences of strings.}  As a result, our cost functions are more complex, and take more time to compute. (See Theorem~\ref{thm:dsssp_is_fixed_parameter_tractable}.)}
  \label{fig:brute-force-algorithm}
\end{figure}

\end{document}